\documentclass[reprint,aps,prl,superscriptaddress,longbibliography]{revtex4-1}
\usepackage[latin9]{inputenc}
\usepackage{color}
\usepackage{times}
\usepackage{amsmath}
\usepackage{amssymb}
\usepackage{stmaryrd}
\usepackage{makecell}
\usepackage[caption=false]{subfig}
\usepackage{graphicx}
\usepackage{tikz-cd}
\usepackage{braket}
\usepackage[unicode=true,
 bookmarks=true,bookmarksnumbered=false,bookmarksopen=false,
 breaklinks=false,pdfborder={0 0 1},backref=false,colorlinks=true]
 {hyperref}
\hypersetup{
 linkcolor=magenta, urlcolor=blue, citecolor=blue, pdfstartview={FitH}, hyperfootnotes=true, unicode=true}

\DeclareMathOperator{\poly}{poly}

\newcommand{\class}[1]{\mathsf{#1}}

\newcommand{\EXP}[0]{\class{EXP}}
\newcommand{\NP}[0]{\class{NP}}

\newcommand{\sz}{\mathrm{Z}}
\newcommand{\sx}{\mathrm{X}}
\newcommand{\sy}{\mathrm{Y}}
\newcommand{\si}{\mathrm{I}}

\newcommand{\tX}{\tilde{X}}
\newcommand{\tZ}{\tilde{Z}}

\providecommand{\U}[1]{\protect\rule{.1in}{.1in}}
\newtheorem{theorem}{Theorem}

\newtheorem{definition}[theorem]{Definition}

\newtheorem{lemma}[theorem]{Lemma}

\newenvironment{proof}[1][Proof]{\noindent\textbf{#1.} }{\ \rule{0.5em}{0.5em}}
\begin{document}
\title{Self-Testing Quantum Error Correcting Codes: Analyzing Computational Hardness}
\author{En-Jui Kuo}
\affiliation{
Physics Division, National Center for Theoretical Sciences, Taipei 106319, Taiwan }

\author{Li-Yi Hsu}
\affiliation{
Physics Division, National Center for Theoretical Sciences, Taipei 106319, Taiwan }
\affiliation{Department of Physics, Chung Yuan Christian University, Chungli 32081, Taiwan}

\begin{abstract}
We present a generalization of the tilted Bell inequality for quantum $[[n,k,d]]$ error-correcting codes and explicitly utilize the simplest perfect code, the $[[5,1,3]]$ code, the Steane $[[7,1,3]]$ code, and Shor's $[[9,1,3]]$ code, to demonstrate the self-testing property of their respective codespaces. Additionally, we establish a framework for the proof of self-testing, as detailed in \cite{baccari2020device}, which can be generalized to the codespace of CSS stabilizers. Our method provides a self-testing scheme for $\cos\theta \ket{\bar{0}}+\sin\theta \ket{\bar{1}}$, where $\theta \in [0, \frac{\pi}{2}]$, and also discusses its experimental application. We also investigate whether such property can be generalized to qudit and show one no-go theorem.

We then define a computational problem called \textit{ISSELFTEST} and describe how this problem formulation can be interpreted as a statement that maximal violation of a specific Bell-type inequality can self-test a particular entanglement subspace. We also discuss the computational complexity of \textit{ISSELFTEST} in comparison to other classical complexity challenges and some related open problems.
\end{abstract}
\maketitle

{\it Introduction.---} 

Quantum self-testing is a significant concept in quantum information theory, focusing on the device-independent characterization of quantum systems. In simple terms, it allows researchers to infer details about quantum systems and experiments without detailed knowledge of the underlying physics, essentially treating the systems as "black boxes." At the core of quantum self-testing is the utilization of Bell nonlocality \cite{brunner2014bell, scarani2019bell}, a phenomenon first identified by John Bell in 1964. Bell nonlocality is central to understanding quantum entanglement, where Bell inequalities are used to test the nonlocal nature of quantum systems. Essentially, these inequalities provide a benchmark that quantum systems can surpass, but classical systems cannot. This violation of Bell inequalities by entangled quantum systems serves as a foundation for self-testing methodologies. The best-known example of a Bell inequality is the Clauser-Horne-Shimony-Holt (CHSH) inequality:
\begin{equation}
I_0=A_0B_0+A_0B_1+A_1B_0-A_1B_1.
\end{equation}
The maximal value of the Bell inequality can be shown using a beautiful one-line proof with the method called Sum of Squares (SOS) Techniques \cite{barak2014sum}
\begin{equation}
    \left(\frac{A_0+A_1}{\sqrt{2}}-B_0\right)^2+\left(\frac{A_0-A_1}{\sqrt{2}}-B_1\right)^2=4-\sqrt{2}I_{0}\geq 0.
\end{equation}
Its maximal quantum value is $2\sqrt{2}$ and is achieved by the maximally entangled state of two qubits $\frac{1}{\sqrt{2}}(\ket{00}+\ket{11})$ and also $B_{x}=\frac{A_0+(-1)^{x}A_1}{\sqrt{2}}.$ The reverse of this statement is the famous Tsirelson's inequality and Exact CHSH rigidity \cite{summers1987maximal, mckague2012robust}. An approximate or quantitative version of CHSH rigidity was obtained by \cite{mckague2012robust}, who stated that if you have a quantum value of $2\sqrt{2}-\epsilon$ for some small $\epsilon>0$, then there exist isometries under which the entangled state between two parties is $O(\sqrt{\epsilon})$ close to the Bell pair up to local rotation and isometry. A very good review on this topic is "Quantum multiplayer games, testing, and rigidity" by Tomas Vidick \cite{vidick2024}.

Following the original Bell inequality, many variants of Bell inequalities have been developed for qubits, qudits, and even continuous variables, including the Collins-Gisin inequality \cite{collins2004relevant} and the Collins-Gisin-Linden-Massar-Popescu (CGLMP) inequality \cite{collins2002bell, masanes2002tight, fonseca2018survey}, as well as its infinite-dimensional version \cite{zohren2008maximal}, symmetric Bell inequalities \cite{bancal2010looking}, and Bell inequalities for continuous variables \cite{cavalcanti2007bell, he2010bell, chen2002maximal}, among others.

However, not all these various Bell inequalities can achieve the robustness of maximal violation compared to the original CHSH inequality, even in the multi-party qubit case. Thus, the notion of self-testing arises. \cite{baccari2020device} defines the self-testing of an entangled subspace, not just for a single state, and uses $[[5,1,3]]$ codes and toric codes. Several works have followed, including those on graph states \cite{baccari2020scalable}.

More abstractly, one can define the class MIP*, which represents multi-prover interactive proofs where the provers are quantumly entangled, and it is equivalent to RE, the class of recursively enumerable languages. The equivalence MIP* $=$ RE \cite{ji2021mip, mousavi2020complexity} indicates that certain computational problems, previously considered unverifiable in a classical setting, can indeed be verified given quantum resources. This is closely related to the concept of self-testing quantum states, where a quantum system's state can be verified through its behavior in a specific test. However, even small noise will cause such complexity to collapse to the classical result MIP $=$ NEXP \cite{dong2023computational}.

{\it Preliminaries and Notations.---}
In this paper, we use $\sx, \sy, \sz, \si$ to represent the standard Pauli matrices, where 
\[
\sx ={\begin{pmatrix}0&1\\1&0\end{pmatrix}}, \quad \sy ={\begin{pmatrix}0&-i\\i&0\end{pmatrix}}, \quad \sz ={\begin{pmatrix}1&0\\0&-1\end{pmatrix}}.
\]
$\si$ serves as the identity operator in any dimension. We use $\mathbb{C}$ to denote the field of complex numbers. The set of integers from $1$ to $n$ is denoted by $[n].$

For each $n$-party system, denoted by $A_0^{i}$ and $A_1^{i}, \forall 1\leq i \leq n$, there exist two Hermitian operators. The only constraint is that the operators $A_x^i, x \in \{0,1 \}$ can only have eigenvalues $\pm 1$, implying that their squares are identity operators, i.e., $(A_0^{i})^2=\si=(A_1^{i})^2$ for all $i \leq n$. Since $(A_x^i)^2$ equals the identity, the minimal polynomial of $A_x^i$ is completely reducible over $\mathbb{C}$, which implies that $A_x^i$ is always diagonalizable \cite{barrera2023linear}.

{\it Bell Inequality.---}
Now, let us introduce a general expression for an $n$-multipartite Bell inequality $L$ as in \cite{baccari2020device}:
{\scriptsize
\begin{align}\label{eq:bel}
\sum_{k=1, 1 \leq i_1<i_2<..<i_k\leq N, x_{i_1}, x_{i_2},..,x_{i_N}=0,1}^{N}\alpha_{x_{i_1},x_{i_2},...,x_{i_k}}^{i_1,i_2,..., i_k}A_{x_{i_1}}^{i_1}A_{x_{i_2}}^{i_2}...A_{x_{i_k}}^{i_k}.
\end{align}
}
Even more generally, since $A$'s may not commute in a single party, we can have terms including $A^{i}_0A^{i}_1A^{i}_0$ or even $A^{i}_0A^{i}_1A^{i}_0A^{i}_1$. These terms play an important role in our construction and are sometimes referred to as monomials \cite{barizien2024custom}. We say $L$ \eqref{eq:bel} achieves a maximal violation $c$ over all possible $A$'s if for all $\ket{\psi}$, we have $\braket{L}= \braket{\psi|L|\psi}\leq c.$

\begin{figure}[t]
    \centering
\includegraphics[width=0.85\columnwidth]{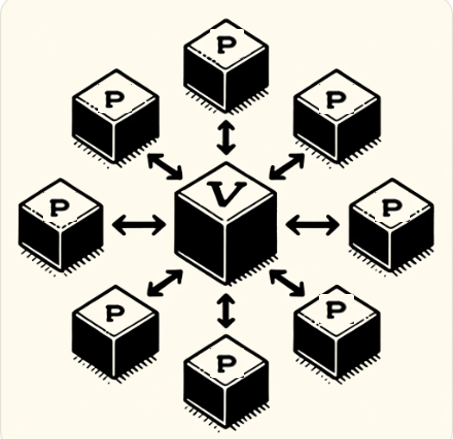}
    \caption{We describe the standard model of self-testing with multiple interactive provers. A classical verifier $\textbf{V}$ aims to certify a specific quantum property of the resource shared by non-communicating provers $\textbf{P}$, potentially quantum-entangled. The verifier communicates with the provers using classical messages (denoted by brown arrows), and the provers respond with classical outputs. Nonlocal correlations between these outputs enable the verifier to draw significant conclusions about the quantum property.}
    \label{fig:main}
\end{figure}

We now introduce the Stabilizer codes \cite{gottesman1997stabilizer}. The $N$-fold tensor products of the
Pauli operators $\{\sx, \sz, \sy, \si \}$ with the overall factor $\pm 1$ or $ \pm i$ form the Pauli group $\mathbb{P}_N$ under matrix multiplication. A
stabilizer $\mathbb{S}_N$ is any Abelian subgroup of $\mathbb{P}_N$, and the stabilizer coding space $C_N$ consists of all vectors $\ket{\psi}$ such that $S_i\ket{\psi}=\ket{\psi} \forall S_i \in \mathbb{S}_N.$ We denote 
${\displaystyle \left[[n,k, d\right]]}$ stabilizer quantum error-correcting code to encode $k$ logical qubits into 
$n$ physical qubits. $d$ is the maximal weight that one can correct.
Usually, we specify particular subset of $\mathbb{S}_N$ which called generators of stabilizer.

Here we use $S$ to denote the stabilizer code composed of Pauli strings. For example, in $[[5,1,3]]$ codes, we have four generators. We use $\tilde{S}$ to denote the operator composed of general operators $X, Y, Z, I$ and those composed of $A$ from each party.

\begin{definition} \cite{baccari2020device}
The behavior $\mathcal{P}$ self-tests the entangled subspace spanned by the set of entangled states $\{\ket{\psi_i}\}_{i=1}^k$ if for any pure state $\ket{\tilde{\psi}}_{PE}\in\mathcal{H}_{PE}$ compatible with $\mathcal{P}$ through $L$ one can deduce that: i) every local Hilbert space has the form $\mathcal{H}_{P}=\mathcal{H}_{P'}\otimes \mathcal{H}_{P''}$;
ii) there exists a local unitary transformation $U_P = U_1\otimes\cdots\otimes  U_N$ acting on $ \mathcal{H}_P $ such that
\begin{align*}
   U_P \otimes \mathbb{I}_E(\ket{\tilde{\psi}}_{PE}) = \sum_{i} {c_i}\ket{\psi_i}_{P'}\otimes\ket{\xi_i}_{P''E}
\end{align*}
for some normalised states $\ket{\xi_{i}}\in \mathcal{H}_{P''}\otimes \mathcal{H}_E$ and some nonnegative numbers $c_i\geq 0$ such that $\sum_i c_i^2 = 1$.    
\end{definition}
In this paper, we focus on two particular cases of such definition. First, we want to certify a particular 2-dimensional codespace of stabilizer codes. So in this case. We need to show there exists $U_P$ and $\ket{\tilde{\psi}}_{PE}$ satisfy
\begin{equation} 
    U_P \otimes \mathbb{I}_E\ket{\tilde{\psi}}=c\ket{\psi_1}\otimes \ket{\xi_1}+\sqrt{1-c^2}\ket{\psi_2}\otimes \ket{\xi_2}.
\end{equation}
where $\ket{\psi_1}, \ket{\psi_2}$ are two orthonormal states spanning the code space and $c \in [0,1].$ In the extreme case, when we want to certify a particular code state $\ket{\psi_1}$ above, equivalently, we need to show $ U_P \otimes \mathbb{I}_E(\ket{\tilde{\psi}})=\ket{\psi_1}\otimes \ket{\xi_1}.$ Here $\ket{\xi_1}$ is a auxiliary state.

In the following, we will first review the standard tilted Bell inequality on $\cos\theta \ket{00}+\sin \theta \ket{11}$. We also recall that the self-test of such a state $\cos\theta \ket{0}^n+\sin \theta \ket{1}^n$ has been derived \cite{baccari2020scalable}. Next, we generalize the method of \cite{baccari2020device} on $[[5,1,3]]$ codes and show how to generalize it to $\cos\theta \ket{\bar{0}}+\sin \theta \ket{\bar{1}}.$ After that, we move to the first CSS code called Steane code $[[7,1,3]]$ and provide the details on $\cos\theta \ket{\bar{0}}+\sin \theta \ket{\bar{1}}.$ Finally, we move to the Shor's $[[9,1,3]]$ code $\cos\theta \ket{\bar{0}}+\sin \theta \ket{\bar{1}}.$

In summary, we define the computational problem called \textit{ISSELFTEST} and describe how this problem formulation can be interpreted as a statement that maximal violation of a specific Bell-type inequality can self-test a particular entanglement subspace. We also discuss the computational complexity of \textit{ISSELFTEST} in comparison to other classical complexity challenges and some related open problems.

{\it Stabilizer codes using $[[5,1,3]]$ as example .---}

We now consider $[[5,1,3]]$ codes as our first example which is a perfect code and the simplest non-CSS code \cite{eczoo_qecc}
with the following stabilizers 
\begin{align*}
S_1=\sx_1\sz_2\sz_3\sx_4\si_5,S_2=\si_1\sx_2\sz_3\sz_4\sx_5,\nonumber \\
S_3=\sx_1\si_2\sx_3\sz_4\sz_5,S_4=\sz_1\sx_2\si_3\sx_4\sz_5    
\end{align*}  
We will modify the proof \cite{baccari2020device} and show how to test particular codewords. The code space denoted $C_5$ which is a two-dimensional space forming from $\ket{\bar{0}}, \ket{\bar{1}}$.
The logical operator is ${\displaystyle {\bar {\sx}}=\sx_1 \sx_2 \sx_3 \sx_4 \sx_5},  {\displaystyle {\bar {\sz}}=\sz_1 \sz_2 \sz_3 \sz_4 \sz_5}.$

We define $X, Z$ by local observable $A_x^i$.
$X_{1}'= \frac{A_0^1+A_1^1}{2 \cos \mu}, Z_{1}'= \frac{A_0^1-A_1^1}{2 \sin \mu}$ while for the remaining parties we simply setting $X_i= A_0^i, Z_i = A_1^i$ for $i=2,3,4$ and $X_1=\frac{X_1'}{|X_1'|}, Z_1=\frac{Z_1'}{|Z_1'|}.$
We set the following $4$ operators $\tilde{S}_1=X_1Z_2Z_3X_4, \tilde{S}_2=X_2Z_3Z_4X_5, \tilde{S}_3=X_3Z_4Z_5X_1, \tilde{S}_4=X_4Z_5Z_1X_2,
\tilde{X}=X_1X_2X_3X_4X_5,
\tilde{Z}=Z_1Z_2Z_3Z_4Z_5$ We have $\tilde{S}_1^2=X_1^2, \tilde{S}_2^2=1, \tilde{S}_3^2=X_1^2, \tilde{S}_4^2=Z_1^2, \tilde{X}^2=X_1^2, \tilde{Z}^2=Z_1^2.$
Also, $\tilde{P}=\cos2\theta \tilde{Z}+\sin2\theta\tilde{X}, \tilde{P}^2=\cos^2(2\theta)Z_1^2+\sin^2(2\theta)X_1^2+\sin(2\theta)\cos(2\theta)\{\tilde{X}, \tilde{Z}\}.$ We can pick $\mu=\pi/4.$

Now we can use SOS method and construct the following: Let us write the following trivial inequality: 
$\alpha_0(\tilde{P}-I)^{2}+\sum_{i=1}^{4}\alpha_i(\tilde{S}_i-I)^{2}\succeq0$ implies 
\begin{eqnarray}\label{eq:i5}
    \alpha_0 \tilde{P}^2+\sum_{i=1}^{4}\alpha_i \tilde{S}_i^2-2(\alpha_0 \tilde{P}+\sum_{i=1}^{4}\alpha_i \tilde{S}_i)+\sum_{i=0}^{4}\alpha_i\succeq0
\end{eqnarray}
If we can pick particular $\alpha_i, i=1,2,3,4$ such that $\sum_{i=1}^{4}\alpha_i \tilde{S}_i^2=\sum_{i=1}^{4}\alpha_i.$ Indeed,  \cite{baccari2020device} chooses such coefficients satisfy it. In such case:
 We have imply
\begin{equation}\label{eq:i5}
    -\alpha_0\tilde{P}^2+2\alpha_0 \tilde{P}+2\sum_{i=1}^{4}\alpha_i \tilde{S}_i \leq \alpha_0+2\sum_{i=1}^{4}\alpha_i.
\end{equation}
We can let $I_{5}'=-\alpha_0\tilde{P}^2+2\alpha_0 \tilde{P}+2\sum_{i=1}^{4}\alpha_i \tilde{S}_i$ which eq \ref{eq:i5} shows $\langle I_{5}' \rangle \leq \alpha_0+2\sum_{i=1}^{4}\alpha_{i}$ for any state $\ket{\psi}$ living in any finite Hilbert space.  Abstractly, after writing up to this stage, expanding \( I_5' \) in terms of \( A_i^j \) and setting \(\alpha_0 = 0\) will recover the Bell inequality derived in \cite{baccari2020device}.

Now we can prove the so-called self-test property proposed by \cite{baccari2020device, baccari2020scalable} when we achieve the maximal violation of $I_5'$. From SOS decomposition and $\alpha_i>0, \forall i=0,1,2,3,4$ shows $\tilde{S}_i\ket{\psi}=\ket{\psi} i \in [4].$ 
Up to now, we can use  Theorem \ref{thm:513} in appendix \ref{app:ap2} to show that these $X, Z$ satisfy $X^2=Z^2=1, XZ=-ZX$. Then utilize Lemma \ref{lem:qubit} which shows that there exists a basis such that $X, Z$ indeed has commutator relation as standard $\sx, \sz.$ There exists an unitary $U$ such that $U X U^{\dagger}=\sx \otimes \si, U Z U^{\dagger}=\sz \otimes \si$. 
Utilizing this lemma and result in appendix \ref{app:ap3} which will force the $\ket{\psi}=(a \ket{\bar{0}}\ket{\xi_1}+b \ket{\bar{1}}\ket{\xi_1})$ where $\ket{\xi_1}, \ket{\xi_2}$ are some auxiliary states. 

Now we go back to the relation $\tilde{P}\ket{\psi}=\ket{\psi}.$ Now $\tilde{P}$ indeed becomes the projector since $\tilde{P}^2=\cos^2(2\theta)Z_1^2+\sin^2(2\theta)X_1^2+\sin(2\theta)\cos(2\theta)\{\tilde{X}, \tilde{Z}\}=1.$ $\tilde{P}\ket{\psi}=\ket{\psi}$ implies that $a \cos 2 \theta + b \sin 2 \theta = a, -b \cos 2 \theta + a \sin 2 \theta = b.$ The nontrivial $a, b$ implies $a/b=\cot \theta.$ If we restrict angle $\theta \in [0, \pi/2],$ which forces $a=\cos \theta,  b = \sin \theta.$ Now we finish the proof. We restate in the following:
\begin{theorem}
$I_5'$ can certify $\cos \theta \ket{\bar{0}}+\sin \theta \ket{\bar{1}}$ by choosing $\alpha_i>0, i=0, 1,2,3,4.$ If we pick $\alpha_0=0, \alpha_i>0, i= 1,2,3,4.$ then $I_5=I_5'|_{\alpha_0=0}$ certify the code space.
\end{theorem}
The complete expression of $I_5'$ is very lengthy, so we leave it in the Appendix \ref{app:ap2}.

Even though the $[[5,1,3]]$ code is the simplest perfect code, its self-test proof is fairly nontrivial \cite{baccari2020device}. After completing the analysis for the $5$-qubit case, we now turn our attention to other stabilizers. This proof does not easily extend to other CSS codes, as will be illustrated with the $[[9,1,3]]$ code, where we address the complexities involved. Additionally, as shown in Theorem \ref{thm:q513}, this proof cannot be directly applied when the qubit is generalized to a qudit with dimension \( q \) (where \( q > 2 \)), indicating that it is not easily adaptable to qudit codes or other types of stabilizer codes.

{\it $[[9,1,3]]$ codes .---}
We now describe another scenario by simplest Shor
$[[9, 1, 3]]$ codes which is the first quantum error-correcting code. These are generators.
\begin{align*}
    S_1 &= \sz_1\sz_2, & S_2 &= \sz_1\sz_3 \nonumber \\ 
    S_3 &= \sz_4\sz_5, & S_4 &= \sz_4\sz_6 \nonumber \\ 
    S_5 &= \sz_7\sz_8, & S_6 &= \sz_7\sz_9 \nonumber \\
    S_7 &= \sx_1\sx_2\sx_3\sx_4\sx_5\sx_6, &
    S_8 &= \sx_1\sx_2\sx_3\sx_7\sx_8\sx_9.
\end{align*}
The logical operator is $\bar {\sx}=\prod_{i=1}^{9}\sx_i,  \bar {\sz}=\prod_{i=1}^{9}\sz_i.$

We set the following $X,Z$:
$X_{k}'= \frac{A_0^1+A_1^1}{2 \cos \mu_k}, Z_{k}'= \frac{A_0^1-A_1^1}{2 \sin \mu_k}, k \in \{1,4,7.\}$ and 
$X_k=\frac{X_k'}{|X_k'|}, Z_k=\frac{Z_k'}{|Z_k'|}.$
While for the remaining parties, we simply setting $X_k= A_0^k, Z_k = A_1^k$ for $k \in \{2,3,5,6,8,9\}. $ We set $\mu_k=\pi/4. $ We then define a similar object:
\begin{align*}
    \tilde{S}_1 &= Z_1Z_2,  & \tilde{S}_2 &= Z_1Z_3 \nonumber \\ 
    \tilde{S}_3 &= Z_4Z_5,  & \tilde{S}_4 &= Z_4Z_6 \nonumber \\ 
    \tilde{S}_5 &= Z_7Z_8,  & \tilde{S}_6 &= Z_7Z_9 \nonumber \\
    \tilde{S}_7 &= X_1X_2X_3X_4X_5X_6, & \tilde{S}_8 &= X_1X_2X_3X_7X_8X_9 
\end{align*}
Here we define $\tilde{S}_9=X_4X_5X_6X_7X_8X_9$ which is critical. Here one may think it is redundant to add this $\tilde{S}_9$ since $S_9=S_7S_8$. However, notice that these $\tilde{S}$ are not stabilizers yet since these $X,Z$ does not complete following the rules $X^2=Z^2=1, XZ=-ZX.$ To be more precisely, if we demand $\tilde{S}_9\ket{\psi}=\ket{\psi}$ implying $X_4X_5X_6X_7X_8X_9\ket{\psi}=\ket{\psi}.$ On the other hand, even $\tilde{S}_7\ket{\psi}=\ket{\psi}, \tilde{S}_8\ket{\psi}=\ket{\psi}$ can only imply $\tilde{S}_7\tilde{S}_8\ket{\psi}=X_1^2X_2^2X_3^2X_4X_5X_6X_7X_8X_9\ket{\psi}=\ket{\psi}.$ These two are different since we have not yet demanded $X^2=1.$
We define $\tilde{P}=\cos2\theta \tilde{Z}+\sin2\theta\tilde{X}$ as before. So here there are main difference compared to $[[5,1,3]]$. First,  we have three antisymmetric pairs not just one $1.$ Second, there is add additional constraint. The reason will be clear once \ref{thm:913} is estabilised.

Let us write down the following SOS expression:
$\alpha_0(\tilde{P}-I)^{2}+\sum_{i=1}^{9}\alpha_i(\tilde{S}_i-I)^{2}\succeq0.$  We define this as $I_9'$. 
Now we can express our results. The complete expression of $I_9'$ is very lengthy, so we leave it in the Appendix \ref{app:ap2}.
\begin{theorem}
$I_9'$ can certify $\ket{\psi}=\cos \theta \ket{\bar{0}}+\sin \theta \ket{\bar{1}}$ by choosing $\alpha_i>0, i \in \{0\} \cup [9]$ If we pick $\alpha_0=0, \alpha_i>0, i \in [9]$ then $I_9=I_9'|_{\alpha_0=0}$ certify the codespace.
\end{theorem}
\begin{proof}
From SOS decomposition, if $I_9'$ achieves the maximal value. Then we must have $\tilde{S}_k\ket{\psi}=\ket{\psi} k \in [9], \Tilde{P}\ket{\psi}=\ket{\psi}.$ The $\tilde{S}$ parts force $\tilde{P}$ be the projector $\tilde{P}^2=\cos^2(2\theta)Z_1^2+\sin^2(2\theta)X_1^2+\sin(2\theta)\cos(2\theta)\{\tilde{X}, \tilde{Z}\}=1$. We then first use theorem \ref{thm:913} and then self-test property in Appendix \ref{app:ap3}.
Let us write $\ket{\psi}=a \ket{\bar{0}}\ket{\xi_1}+b\ket{\bar{1}}\ket{\xi_2}$ in particular four dimensional subspace. $\tilde{P}\ket{\psi}=\ket{\psi}$ implies that $a \cos 2 \theta + b \sin 2 \theta = a, -b \cos 2 \theta + a \sin 2 \theta = b.$ The nontrivial $a, b$ implies $a/b=\cot \theta.$ If we restrict angle $\theta \in [0, \pi/2],$ which forces $a=\cos \theta,  b = \sin \theta.$ Now we finish the proof. 
\end{proof}
Notice that the anti-symmetric pair and additional constraint is crucial for us to show Theorem \ref{thm:913}.



{ \it $[[7, 1, 3]]$ codes.--- }

Here we move to our final example. We demonstrate our proof be also used on Steane code \cite{steane1996error} which is a standard famous type of error-correcting code also known as CSS code \cite{eczoo_qecc}. The Steane code has 6 generators. Let us label them using $1,2,3,5,6,7.$ The reason will be clear after.
\[
{\displaystyle {\begin{aligned}
&S_1=\sx_4\sx_5\sx_6\sx_7, \\
&S_2=\sx_2\sx_3\sx_6\sx_7, \\
&S_3=\sx_1\sx_3\sx_5\sx_7, \\
&S_5=\sz_4\sz_5\sz_6\sz_7, \\
&S_6=\sz_2\sz_3\sz_6\sz_7, \\
&S_7=\sz_1\sz_3\sz_5\sz_7
\end{aligned}}}
\]
and ${\displaystyle [[7,1,3]]}$ is the first in the family of quantum Hamming codes. The logical ${\displaystyle \bar{X}}$ and 
${\displaystyle \bar{Z}}$ gates are 
\begin{align}
\bar {\sx}=\prod_{i=1}^{7}\sx_i,  \bar {\sz}=\prod_{i=1}^{7}\sz_i.
\end{align}
We then define the following $X$ and $Z$'s by making the following substitution 
\begin{align*}
   & X_{k} = \frac{(A_0^k+A_1^k)}{2\cos \mu_k}, Z_{k} = \frac{(A_0^k-A_1^k)}{2\sin \mu_k}, k \in \{2,3,5,7\} \nonumber \\
& X_{k} = A_0^{k}, Z_{k} = A_1^{k}, k \in \{1,4,6 \}.
\end{align*}
We still pick $\mu_k=\pi/4.$
We set the following $6$ operators $\tilde{S}_1=X_4X_5X_6X_7, \tilde{S}_2=X_2X_3X_6X_7, \tilde{S}_3=X_1X_3X_5X_7, \tilde{S}_5=Z_4Z_5Z_6Z_7, \tilde{S}_6=Z_2Z_3Z_6Z_7, \tilde{S}_7=Z_1Z_3Z_5Z_7
$
as well as two logical operators. $
\tilde{X}=\prod_{i=1}^{7}X_i,
\tilde{Z}=\prod_{i=1}^{7}Z_i.$

As before, to process this proof, we add two additional operators. 
\begin{align*}
    \tilde{S}_4=X_1X_2X_5X_6, \tilde{S}_8=Z_1Z_2Z_5Z_6.
\end{align*}
Let us write down the following expression:
$\alpha_0(\tilde{P}-I)^{2}+\sum_{i=1}^{8}\alpha_i(\tilde{S}_i-I)^{2}\succeq0$ 
. We define this as $I_7'$. 
Now we can express our results. The complete expression of $I_7'$ is very lengthy, so we leave it in the Appendix \ref{app:ap2}.


\begin{theorem}
$I_7'$ can certify $\cos \theta \ket{\bar{0}}+\sin \theta \ket{\bar{1}}$ by choosing $\alpha_i>0, i \in \{0\} \cup [8]$ If we pick $\alpha_0=0, \alpha_i>0, i \in [8]$ then $I_7=I_7'|_{\alpha_0=0}$ certified the codespace.
\end{theorem}

\begin{proof}
From SOS decomposition, if $I_7'$ achieves the maximal value. Then we must have $\tilde{S}_k\ket{\psi}=\ket{\psi} k \in [8], \Tilde{P}\ket{\psi}=\ket{\psi}.$ If we first focus on the $\tilde{S}_k\ket{\psi}=\ket{\psi} k \in [8].$ Using the similar approach in theorem \ref{thm:713} and Appendix \ref{app:ap3}, we can make sure our state $\ket{\tilde{\psi}}=c\ket{\psi_1}\otimes \ket{\xi_1}+\sqrt{1-c^2}\ket{\psi_2}\otimes \ket{\xi_2}$ and $\ket{\psi_1}, \ket{\psi_2}$ are orthonormal state of codespace. We can pick $\ket{\psi_1}=\ket{\bar{0}}, \ket{\psi_2}=\ket{\bar{1}}.$ As now $\tilde{P}$ indeed becomes the projector $\tilde{P}^2=\cos^2(2\theta)Z_1^2+\sin^2(2\theta)X_1^2+\sin(2\theta)\cos(2\theta)\{\tilde{X}, \tilde{Z}\}=1$ and we also have  $U\tilde{P}U^{\dagger}=(\cos 2\theta \tilde{\sz}+ \sin 2\theta \tilde{\sx}) \otimes \si.$ 
We have $\ket{\psi}=a \ket{\bar{0}}\ket{\xi_1}+b\ket{\bar{1}}\ket{\xi_2}$ in particular two dimensional subspace. $\tilde{P}\ket{\psi}=\ket{\psi}$ implies that $a \cos 2 \theta + b \sin 2 \theta = a, -b \cos 2 \theta + a \sin 2 \theta = b.$ The nontrivial $a, b$ implies $a/b=\cot \theta.$ If we restrict angle $\theta \in [0, \pi/2],$ which forces $a=\cos \theta,  b = \sin \theta.$ Now we finish the proof. The nontrivial part is to design suitable $\tilde{S}$ and find the anticommutating pair such that theorem \ref{thm:713} can be true. Notice that in Appendix \ref{app:app4} we discuss some related works \cite{makuta2021self}.
\end{proof}

{\it General Framework.---}
Now we formalize the general framework of certifying some codewords of the stabilizer codes $\mathbb{S}_N.$

\begin{definition}\textit{ISSELFTEST}\label{def:is}
Given a set of operators $\tilde{S}_k$ such that $\tilde{S}_k \ket{\psi} = \ket{\psi}$, we want to determine whether (or find) there exists a subset $A \subset [n]$ such that the following conditions are met:
\begin{align*}
    \begin{Bmatrix}
\{X_k, Z_k\}=0, k \in A\\
\tilde{S}_k\ket{\psi}=\ket{\psi},\forall \tilde{S}_k \in S.\\
X_k^2=Z_k^2=1,\forall k \in [n]/A
\end{Bmatrix} \Rightarrow  \begin{Bmatrix}
X_k^2\ket{\psi}=Z_k^2\ket{\psi}=\ket{\psi},\forall \tilde{S}_k \in S \\
\{X_k, Z_k\}=0, \forall k \geq 1\\
\end{Bmatrix}.
\end{align*}
When the left-hand side has $\poly(n)$ many $\tilde{S}_k$, we call this decision problem \textit{ISSELFTEST}. We can define two versions. The first version requires $\tilde{S}_k$ to be the product of $X, Z$ operators, namely $\tilde{S}_k$ is a projector. In another case, one can define $S_k$ to be the general operator, for example, a tilted GHZ state \cite{baccari2020scalable} and its tilted operator.
\end{definition}
Let us first describe why this definition is useful.
For a given $\tilde{S}_i$ term (expressed in terms of a product of $X, Z$), we make the following substitutions:
\begin{align*}
   & X_{k} = \frac{(A_0^1+A_1^1)}{2\cos \mu}, Z_{k} =  \frac{(A_0^1-A_1^1)}{2\sin \mu}, k \in A \nonumber \\
& X_{k} = A_0^{k}, Z_{k} =  A_1^{k}, k \in [n]/A.
\end{align*}
For $\alpha_i > 0, \forall i$, we can write $\sum_{\tilde{S}_i \in S} \alpha_i (\tilde{S}_i - I)^2 \succeq 0$ by defining the following Bell inequality as $I_n = f(A_0^{k}, A_1^{k}), , k \in [n].$ Previous SOS theory implies $\braket{I_n} \leq c_n$ for some number $c_n > 0.$
If such a violation achieves the maximum, we trivially have $\tilde{S}_k \ket{\psi} = \ket{\psi}, \forall \tilde{S}_k \in S.$ However, definition \ref{def:is} will utilize \ref{lem:qubit} in Appendix \ref{app:ap1}, which allows us to introduce the set of local unitary operations $U_i: \mathcal{H}_{P_i} \to \mathcal{H}_{P_i}$ such that:
\begin{eqnarray}
U_i Z_i U_i^{\dagger} = \sz_i \otimes \si, \quad U_i X_i U_i^{\dagger} = \sx_i \otimes \si
\end{eqnarray}
with $i = 1, 2, ..., n.$
In Appendix \ref{app:ap3}, we then show that this criterion already certifies such a codespace.

To discuss the computational complexity of this problem, we begin by defining the relevant classes. The class \(\EXP\) (exponential time) consists of problems that can be solved by an algorithm in \(O(2^{p(n)})\) time for some polynomial \(p(n)\), while \(\NP\) (nondeterministic polynomial time) includes problems for which a solution can be verified in polynomial time. It is obvious that \textit{ISSELFTEST} $\in \EXP$ because there are only \(2^n\) choices for $A$ that can be enumerated to see whether such a deduction can be derived. 
It is intriguing to consider whether this problem belongs to the $\NP$ or $\NP$-complete class. 
An open question remains whether one can always design such a set of \(\tilde{S}_k\) that can certify any codespace.

{\it Outlook.---}

In this paper, we discuss how to self-test for a particular codespace of error-correcting code and then find a self-test procedure for a specific code state. We use the $[[5,1,3]]$, $[[7,1,3]]$, and $[[9,1,3]]$ codes as our examples. The key technicality lies in imposing additional stabilizer constraints and identifying specific combinations of anti-symmetric pairs to make these constraints work. We conclude by mentioning possible future directions.

One interesting generalization is to investigate whether this self-test method can be applied to qudits. The maximal violation of Bell states for qudits has been derived \cite{sarkar2021self, pauwels2022almost}. One obstacle is that, in the qudit case, the $\sx$ and $\sz$ operators cannot both be Hermitian. Thus, it is intriguing to extend this framework to bosonic codes using the bosonic analog of stabilizer codes \cite{eczoo_qecc}, and also to the XP stabilizer formalism \cite{webster2023xp} or XS formalism \cite{ni2015non}. Additionally, single-party scenarios are worth exploring \cite{metger2021self}.

It is also interesting to consider the robustness of the self-test procedure on noisy quantum platforms (NISQ) \cite{preskill2018quantum, bharti2022noisy}. NISQ devices can still be useful for certain types of quantum algorithms, such as the Variational Quantum Eigensolver (VQE) \cite{kandala2017hardware} and Quantum Optimization Algorithms (QAOA) \cite{streif2019comparison}. Standard NISQ platforms include superconducting qubits \cite{wallraff2004strong, krantz2019quantum, larsen2015semiconductor}, trapped ions \cite{bruzewicz2019trapped}, optical systems \cite{o2007optical}, and Rydberg atoms \cite{wu2021concise}. These platforms represent the most promising technologies for building practical quantum computers with a few dozen qubits, paving the way toward quantum advantage. 

Given the increasing prominence of NISQ devices, the study of noisy self-testing states is becoming an important and popular research topic for future exploration.

{\it Acknowledgments.---} 
EJK and LYH thank the National Center for Theoretical Sciences in Taiwan for their support (EJK: 112-2124-M-002-003). 
LYH acknowledges financial support from the National Science and Technology Council in Taiwan.

\bibliography{sample}

\newpage
\clearpage

\section*{SUPPLEMENTAL MATERIAL}

\appendix

\appendix
\section{}\label{app:ap1}
Here we provide full proofs of the self-testing statements for the subspaces considered in the main text. For completeness, we first recall a very useful fact, proven already in Refs. \cite{popescu1992states,kaniewski2016analytic}, that is used in our proofs.

\begin{lemma}\label{lem:qubit}
\cite{popescu1992states,kaniewski2016analytic} Consider two Hermitian operators $\tX$ and $\tZ$ acting on a Hilbert space $\mathcal{H}$ of dimension $D < \infty$ and satisfying the idempotency property $\tX^2 = \tZ^2 = \emph{\openone}$, as well as the anticommutation relation $\lbrace \tX , \tZ \rbrace = 0$. Then, $\mathcal{H} = \mathbb{C}^2 \otimes \mathbb{C}^d$ for some $d$ such that $D = 2d$, and there exists a local unitary operator $U$ for which

\begin{align}\label{iden}
U \tX U^\dagger = \sx \otimes \emph{\openone}_d \, , 
\qquad U \tZ U^\dagger = \sz \otimes \emph{\openone}_d \, ,  
\end{align}

where $\sx$ and $\sz$ are the $2 \times 2$ Pauli matrices introduced before and $\emph{\openone}_d$ is the identity matrix acting on the $d$-dimensional auxiliary Hilbert space. We refer this proof to \cite{baccari2020device, baccari2020scalable}.
\end{lemma}

Here we review the self-test property of the $[[5,1,3]]$ code \cite{baccari2020device}.

\begin{theorem}\label{thm:513}
\begin{align*}
    \begin{Bmatrix}
    \{X_1, Z_1\}=0\\
    \tilde{S}_k\ket{\psi}=\ket{\psi},\forall k \in [4].\\
    X_k^2=Z_k^2=1,\forall 2 \leq k \leq 5
    \end{Bmatrix} \Rightarrow  \begin{Bmatrix}
    X_k^2=Z_k^2=1,\forall k \geq 1 \\
    \{X_k, Z_k\}=0, \forall k \geq 1\\
    \end{Bmatrix}
\end{align*}
\end{theorem}

\begin{proof}
First, we show $X_1^2=Z_1^2=1$ on the support of the state. Namely, $X_1^2\ket{\psi}=Z_1^2\ket{\psi}=\ket{\psi}$. To prove that they also square to identity, we use the definitions for $\tilde{S}_1$ and $\tilde{S}_4$, obtaining
\begin{equation}
\tilde{S}_1^2\ket{\psi}=\tilde{S}_4^2\ket{\psi}=\ket{\psi},
\end{equation}
which, due to the fact that $X_k^2=Z_k^2=1$ for $k=2,\ldots,4$, immediately implies that $X_1^2\ket{\psi}=Z_1^2\ket{\psi}=\ket{\psi}$.

The second step is to show $\{X_k, Z_k\}=0 \forall k \geq 2$. Notice that, by definition, they already satisfy $X_k^2=Z_k^2=1$. To this end, we rewrite the conditions for $k = 1,4$ as

\begin{align}
    X_2 \ket{\psi} &= Z_1 X_4 Z_5 \ket{\psi}, \nonumber \\
    Z_2 \ket{\psi} &= X_1 Z_3 X_4 \ket{\psi},
\end{align}
which leads us to
\begin{equation}
    \lbrace X_2 , Z_2 \rbrace \ket{\psi} = \lbrace X_1 , Z_1 \rbrace Z_3 Z_5 \ket{\psi} = 0,
\end{equation}
where the last equality is a consequence of the anticommutation of $X_1, Z_1$. In a similar way, one can exploit the operator relations stemming from Eq. (4) to obtain anticommutation relations for the remaining three sites. Indeed, we can combine the conditions for $\tilde{S}_3$ and $\tilde{S}_4$ to get
\begin{equation}
 \lbrace X_4 , Z_4 \rbrace \ket{\psi} = \lbrace X_1 , Z_1 \rbrace X_2 X_3 \ket{\psi} = 0 \, ,
\end{equation}
and combine $\tilde{S}_2$ and $\tilde{S}_4$ to obtain
\begin{equation}
 \lbrace X_5 , Z_5 \rbrace \ket{\psi} = \lbrace X_4 , Z_4 \rbrace Z_1 Z_3 \ket{\psi} = 0 \, .
\end{equation}
Finally, we combine the conditions arising from $\tilde{S}_2$ and $\tilde{S}_3$ to show
\begin{equation}
 \lbrace X_3 , Z_3 \rbrace \ket{\psi} = \lbrace X_5 , Z_5 \rbrace Z_1 Z_2 \ket{\psi} = 0 \, .
\end{equation}
\end{proof}

\begin{theorem}\label{thm:913}
\begin{align*}
    \begin{Bmatrix}
    \{X_k, Z_k\}=0, & \, k \in \{1,4,7\} \\
    \tilde{S}_k\ket{\psi}=\ket{\psi}, & \, \forall 1 \leq k \leq 9 \\
    X_k^2=Z_k^2=1, & \, \forall k \in \{2,3,5,6,8,9\} 
    \end{Bmatrix} \Rightarrow \begin{Bmatrix}
    X_k^2=Z_k^2=1, & \, \forall k \\
    \{X_k, Z_k\}=0, & \, \forall k
    \end{Bmatrix}
\end{align*}    
\end{theorem}

\begin{proof}
First, using $\tilde{S}_i^2\ket{\psi}=\ket{\psi}$ for all $i \leq 9$ implies $Z_1^2\ket{\psi}=Z_4^2\ket{\psi}=Z_7^2\ket{\psi}=\ket{\psi}$. By considering $(\tilde{S}_1, \tilde{S}_7)$ and $(\tilde{S}_2, \tilde{S}_7)$, we have
\begin{align*}
& Z_1\ket{\psi}=Z_2\ket{\psi}, \quad X_1\ket{\psi}=X_2X_3X_4X_5X_6\ket{\psi} \\
    \Rightarrow & \{Z_1, X_1 \}\ket{\psi} = \{Z_2, X_2 \} X_3X_4X_5X_6\ket{\psi} \\
    \Rightarrow & \{Z_2, X_2 \} =0.
\end{align*}
A similar proof gives $\{Z_3, X_3\} = 0$. Using the anti-symmetric property, $\{Z_4, X_4\} = 0$ and comparing $(\tilde{S}_3, \tilde{S}_7)$ and $(\tilde{S}_4, \tilde{S}_7)$ imply $\{Z_5, X_5\} = \{Z_6, X_6\} = 0$. Using the anti-symmetric property, $\{Z_7, X_7\} = 0$ and comparing $(\tilde{S}_5, \tilde{S}_8)$ and $(\tilde{S}_6, \tilde{S}_8)$ imply $\{Z_8, X_8\} = \{Z_9, X_9\} = 0$. The remaining part is to show that $X_1^2\ket{\psi} = X_4^2\ket{\psi} = X_7^2\ket{\psi}$.

We now use $\tilde{S}_7^2\ket{\psi}=\tilde{S}_8^2\ket{\psi}=\tilde{S}_9^2\ket{\psi}=\ket{\psi}$ and using $X_k^2=Z_k^2=1$ for all $k \in \{2,3,5,6,8,9\}$, we show $X_1^2X_4^2\ket{\psi}=X_1^2X_7^2\ket{\psi}=X_4^2X_7^2\ket{\psi}=\ket{\psi}$. Finally, this implies that $X_1^4\ket{\psi}=\ket{\psi}$.

Now we appeal to the following claim: if $A$ is a Hermitian operator and $A^4\ket{\psi}=\ket{\psi}$, it implies $A^2\ket{\psi}=\ket{\psi}$. The reason is trivial since $A$ can be diagonalized, so there exists a unitary $U$ such that $A=UDU^{-1}$. However, $A$ is Hermitian and $A^4=1$ implies that the eigenvalue of $A$ on the codespace can only be $\pm 1$. We can see $D^2=1$ implies $A^2\ket{\psi}=\ket{\psi}$. We then have $X_1^2\ket{\psi}=X_4^2\ket{\psi}=X_7^2\ket{\psi}=1$.
\end{proof}

We show the following theorem of the self-test of the Steane code.
\begin{theorem}\label{thm:713}
\begin{align*}
    \begin{Bmatrix}
    \{X_k, Z_k\}=0, & \, k \in \{2,3,5,7\} \\
    \tilde{S}_k\ket{\psi}=\ket{\psi}, & \, \forall 1 \leq k \leq 8 \\
    X_k^2=Z_k^2=1, & \, \forall k \in \{1,4,6\} 
    \end{Bmatrix} \Rightarrow \begin{Bmatrix}
    X_k^2=Z_k^2=1, & \, \forall k \geq 1 \\
    \{X_k, Z_k\}=0, & \, \forall k \geq 1
    \end{Bmatrix}
\end{align*}    
\end{theorem}

Here we add the following two extra operators:
\begin{align*}
    \tilde{S}_4=X_1X_2X_5X_6, \quad \tilde{S}_8=Z_1Z_2Z_5Z_6
\end{align*}

\begin{proof}
We first show that $\{X_k, Z_k\}\ket{\psi}=0$ for all $k \in \{1,4,6\}$. We start by picking $\tilde{S}_3$ and $\tilde{S}_7$ and notice that $X_5Z_5X_7Z_7=Z_5X_5Z_7X_7$.
\begin{align*}
&X_1\ket{\psi}=X_3X_5X_7\ket{\psi}, \quad Z_1\ket{\psi}=Z_3Z_5Z_7\ket{\psi} \\
&\Rightarrow \{X_1, Z_1\}\ket{\psi}=\{X_3, Z_3\}X_5Z_5X_7Z_7\ket{\psi} \\
&\Rightarrow \{X_1, Z_1\}\ket{\psi}=0.
\end{align*}

Next, we pick $\tilde{S}_2$ and $\tilde{S}_6$ and notice that $X_7Z_7X_3Z_3=Z_7X_7Z_3X_3$.
\begin{align*}
&X_6\ket{\psi}=X_2X_3X_7\ket{\psi}, \quad Z_6\ket{\psi}=Z_2Z_3Z_7\ket{\psi} \\
&\Rightarrow \{X_6, Z_6\}\ket{\psi}=\{X_2, Z_2\}X_7Z_7X_3Z_3\ket{\psi} \\
&\Rightarrow \{X_6, Z_6\}\ket{\psi}=0.
\end{align*}

Finally, we pick $\tilde{S}_1$ and $\tilde{S}_5$ and notice that $X_6Z_6X_7Z_7=Z_6X_6Z_7X_7$.
\begin{align*}
&X_4\ket{\psi}=X_5X_6X_7\ket{\psi}, \quad Z_4\ket{\psi}=Z_5Z_6Z_7\ket{\psi} \\
&\Rightarrow \{X_4, Z_4\}\ket{\psi}=\{X_5, Z_5\}X_6Z_6X_7Z_7\ket{\psi} \\
&\Rightarrow \{X_4, Z_4\}\ket{\psi}=0.
\end{align*}

Finally, we want to show $X_2^2=X_3^2=X_5^2=X_7^2=Z_2^2=Z_3^2=Z_5^2=Z_7^2=1$ acting on $\ket{\psi}$. We then use $\tilde{S}_1^2=\tilde{S}_2^2=\tilde{S}_3^2=\tilde{S}_4^2=1$ on $\ket{\psi}$ implying 
\begin{align*}
X_5^2X_7^2\ket{\psi}=X_2^2X_3^2X_7^2\ket{\psi}=X_3^2X_5^2X_7^2\ket{\psi}=X_2^2X_5^2\ket{\psi}=\ket{\psi}.
\end{align*}
We then use a similar method as before to show $X_2^2=X_3^2=X_5^2=X_7^2=Z_2^2=Z_3^2=Z_5^2=Z_7^2=1$.
\end{proof}

In this final section, we will show that the $[[5,1,3]]$ code is very special, as a similar proof cannot hold for the version of the Modular qudit code with $q > 2$ \cite{eczoo_qecc, eczoo_qudit_5_1_3, chau1997five, rains1999nonbinary}. Let us write it very explicitly. In such a case, let us review the stabilizer. The Modular-qudit stabilizer code generalizes the five-qubit perfect code using properties of the multiplicative group $\mathbb{Z}_q$.
It has the following 4 stabilizers:
\begin{align*}
    S_1=\sx_1\sz_2\sz_{3}^{\dagger}\sx_{4}^{\dagger}, \quad S_2=\sx_2\sz_3\sz_{4}^{\dagger}\sx_{5}^{\dagger}, \quad \nonumber\\
    S_3=\sx_3\sz_4\sz_{5}^{\dagger}\sx_{1}^{\dagger}, \quad S_4=\sx_4\sz_5\sz_{1}^{\dagger}\sx_{2}^{\dagger}.
\end{align*}
Notice that $\sz_k\sx_k=\omega_q \sx_k \sz_k$.
We then ask whether a similar theorem to Theorem \ref{thm:513} can hold as when $q=2$. Namely, do we have such a result? Here, $\omega_{q}=e^{2\pi i/q}$ is the $q$-th root of unity, so $\omega_2=-1$.

\begin{theorem}\label{thm:q513}
We will show this is true only when $q=2$. 
\begin{align*}
    \begin{Bmatrix}
    Z_1X_1=\omega_{q}X_1Z_1, \\
    \tilde{S}_k\ket{\psi}=\ket{\psi},  \forall k \in [4], \\
    X_k^q=Z_k^q=1,  2 \leq k \leq 5, \\
    X_k^{\dagger}=X_k^{-1},  Z_k^{\dagger}=Z_k^{-1}, \forall k.
    \end{Bmatrix} \Rightarrow^{?}  \begin{Bmatrix}
    X_k^2=Z_k^2=1, & \forall k, \\
    Z_kX_k=\omega_{q}X_kZ_k, & \forall k.
    \end{Bmatrix}
\end{align*}   
\end{theorem}

\begin{proof}
The proof is elementary. We rewrite the conditions for $k=1,4$.
\begin{align}
    X_2 \ket{\psi} & = Z_1^{\dagger} X_4 Z_5 \ket{\psi}, \nonumber \\
    Z_2 \ket{\psi} & = X_1^{\dagger} Z_3 X_4 \ket{\psi},
\end{align}
which leads us to
\begin{equation}
    (Z_2X_2 -\omega_q^{-1} X_2Z_2) \ket{\psi} = (X_1^{\dagger}Z_1^{\dagger} -\omega_q^{-1}Z_1^{\dagger}X_1^{\dagger})Z_3 Z_5 \ket{\psi} = 0.
\end{equation}
This implies that 
\[
(Z_2X_2 -\omega_q^{-1} X_2Z_2) \ket{\psi} = 0,
\]
which is incompatible with 
\[
(Z_1X_1 -\omega_q X_1Z_1) \ket{\psi} = 0
\]
unless $\omega_q=\omega_q^{-1}$, which is indeed the qubit case, otherwise $\omega_q=1$.
This indicates that such a self-test property is just for qubits, even in the $[[5,1,3]]$ case.
\end{proof}

\section{}\label{app:ap2}
In this section, we manage to write down the Bell inequality as a multi-variable polynomial in terms of the local operator $A$. Notice that $A$ does not need to commute in the local site. Here we will write $I_5'$, $I_7'$, and $I_9'$ explicitly, and in Appendix \ref{app:ap3}, we will finish the final ingredient which shows that the maximal violation of $I_5'$, $I_7'$, and $I_9'$ indeed will make the state the desired state. Now let us proceed with $I_5'$ first since one can compare our example to \cite{baccari2020device} and it is also the simplest error-correcting code.

From Eq. \ref{eq:i5}, we can let $I_{5}'=-\alpha_0\tilde{P}^2+2\alpha_0 \tilde{P}+2\sum_{i=1}^{4}\alpha_i \tilde{S}_i$. Here $\tilde{P}=\cos 2\theta \tilde{Z} + \sin 2\theta \tilde{X}$, and $\tilde{P}^2=\cos^2(2\theta)Z_1^2+\sin^2(2\theta)X_1^2+\sin(2\theta)\cos(2\theta)\{\tilde{X}, \tilde{Z}\}$. Now we have: $\bar{X}=\frac{(A_0^1+A_1^1)}{2 \cos \mu}\prod_{i=2}^{5}A_0^{i}$, $\bar{Z}=\frac{(A_0^1-A_1^1)}{2 \sin \mu}\prod_{i=2}^{5}A_1^{i}$.

We have derived the following expression for $I_5'$:
\begin{align*}
    I_5' &= -\alpha_0  (\sin^2 ( 2\theta) \frac{(A_0^{1} + A_1^{1})^2}{4 \cos^2 \mu } - \cos^2 ( 2\theta) \frac{(A_0^{1} - A_1^{1})^2}{4 \sin^2 \mu} \nonumber \\
    &\quad - \sin(2\theta) \cos(2\theta) \left\{ \frac{(A_0^1 + A_1^1)}{2 \cos \mu} \prod_{i=2}^{5}A_0^{i}, \frac{(A_0^1 - A_1^1)}{2 \sin \mu} \prod_{i=2}^{5}A_0^{i} \right\} \nonumber \\
    &\quad + 2\alpha_0 \left( \cos 2\theta \frac{(A_0^1 - A_1^1)}{2 \sin \mu} \prod_{i=2}^{5}A_1^{i} + \sin 2\theta \frac{(A_0^1 + A_1^1)}{2 \cos \mu} \prod_{i=2}^{5}A_0^{i} \right) \nonumber \\
    &\quad + 2\alpha_1 \frac{A_0^{1} + A_1^{1}}{2 \cos \mu} A_1^{2} A_1^{3} A_0^{4} + 2\alpha_2 A_0^{2} A_1^{3} A_1^{4} A_0^{5} \nonumber \\
    &\quad + 2\alpha_3 \frac{(A_0^{1} + A_1^{1})}{2 \cos \mu} A_0^{3} A_1^{4} A_1^{5} + 2\alpha_4 \frac{(A_0^{1} - A_1^{1})}{2 \sin \mu} A_0^{2} A_0^{4} A_1^{5}.
\end{align*}

If we let $\alpha_0 = 0$, $\alpha_1 = \sqrt{2}$, $\alpha_2 = 1$, $\alpha_3 = \sqrt{2}$, $\alpha_4 = 2\sqrt{2}$, and $\mu = \frac{\pi}{4}$, we recover the Bell inequality from \cite{baccari2020device} (Notice that it is twice compared to \cite{baccari2020device}):
\begin{align}\label{eq:I5}
    I_5 &= 2 (A_0^{1} + A_1^{1}) A_1^{2} A_1^{3} A_0^{4} + 2 A_0^{2} A_1^{3} A_1^{4} A_0^{5} \nonumber \\
        &\quad + 2 (A_0^{1} + A_1^{1}) A_0^{3} A_1^{4} A_1^{5} + 4 (A_0^{1} - A_1^{1}) A_0^{2} A_0^{4} A_1^{5}.
\end{align}

So we indeed show that any behavior achieving the maximal quantum violation of $I_5$ self-tests the entangled code subspace \cite{baccari2020device}.

\paragraph{}

Let us proceed with $I_7'$. We define $\mu = \pi/4$. We can let $I_{7}'=2\alpha_0 \tilde{P}+2\sum_{i=1}^{8}\alpha_i \tilde{S}_i - \sum_{i=1}^{8}\alpha_i \tilde{S}_i^2 - \alpha_0\tilde{P}^2$. Notice here, we need to choose $\{2,3,5,7\}$ as our antisymmetric pairs. So $\tilde{P} = \cos 2\theta \tilde{Z} + \sin 2\theta \tilde{X}$ and $\tilde{P}^2 = \cos^2(2\theta)Z_2^2 Z_3^2 Z_5^2 Z_7^2 + \sin^2(2\theta)X_2^2 X_3^2 X_5^2 X_7^2 + \sin(2\theta)\cos(2\theta)\{\tilde{X}, \tilde{Z}\}$. We recall the stabilizers $\tilde{S}$:

\begin{align*}
&\tilde{S}_1 = X_4 X_5 X_6 X_7, \quad \tilde{S}_2 = X_2 X_3 X_6 X_7, \\
&\tilde{S}_3 = X_1 X_3 X_5 X_7, \quad \tilde{S}_5 = Z_4 Z_5 Z_6 Z_7, \\
&\tilde{S}_6 = Z_2 Z_3 Z_6 Z_7, \quad \tilde{S}_7 = Z_1 Z_3 Z_5 Z_7.
\end{align*}

We add another two operators:

\begin{align*}
    \tilde{S}_4 = X_1 X_2 X_5 X_6, \quad \tilde{S}_8 = Z_1 Z_2 Z_5 Z_6.
\end{align*}

Each $\tilde{X}, \tilde{Z}, X_i, Z_i$ can be expressed in terms of the following rule:

\begin{align*}
   & X_{k} = \frac{(A_0^1 + A_1^1)}{\sqrt{2}}, \quad Z_{k} =  \frac{(A_0^1 - A_1^1)}{\sqrt{2}}, \quad k \in \{2,3,5,7\}, \nonumber \\
   & X_{k} = A_0^{k}, \quad Z_{k} =  A_1^{k}, \quad k \in [7] / \{2,3,5,7\}.
\end{align*}

Plugging all of them inside $I_{7}' = 2\alpha_0 \tilde{P} + 2\sum_{i=1}^{8}\alpha_i \tilde{S}_i - \sum_{i=1}^{8}\alpha_i \tilde{S}_i^2 - \alpha_0 \tilde{P}^2$, we obtain $I_7'$ in terms of various monomials $A$'s.

\paragraph{}
Finally, let us proceed with $I_9'$. We will just sketch the construction since the main idea and previous examples already provide clear illustrations. We define $\mu = \pi/4$. We can let $I_{9}' = 2\alpha_0 \tilde{P} + 2\sum_{i=1}^{9}\alpha_i \tilde{S}_i - \sum_{i=1}^{9}\alpha_i \tilde{S}_i^2 - \alpha_0\tilde{P}^2$. Notice here, we need to choose $\{1,4,7\}$ as our antisymmetric pairs. So $\tilde{P} = \cos 2\theta \tilde{Z} + \sin 2\theta \tilde{X}$ and $\tilde{P}^2 = \cos^2(2\theta) Z_1^2 Z_4^2 Z_7^2 + \sin^2(2\theta) X_1^2 X_4^2 X_7^2 + \sin(2\theta) \cos(2\theta) \{\tilde{X}, \tilde{Z}\}$. Now, the following are our $\tilde{S}$:

\begin{align*}
    &\tilde{S}_1 = Z_1 Z_2, \quad \tilde{S}_2 = Z_1 Z_3, \\
    &\tilde{S}_3 = Z_4 Z_5, \quad \tilde{S}_4 = Z_4 Z_6, \\
    &\tilde{S}_5 = Z_7 Z_8, \quad \tilde{S}_6 = Z_7 Z_9, \\
    &\tilde{S}_7 = X_1 X_2 X_3 X_4 X_5 X_6, \quad \tilde{S}_8 = X_1 X_2 X_3 X_7 X_8 X_9,
\end{align*}

Here we define $\tilde{S}_9 = X_4 X_5 X_6 X_7 X_8 X_9$, which is critical. Finally, each $\tilde{X}, \tilde{Z}, X_i, Z_i$ can be expressed in terms of the following rule:

\begin{align*}
    &X_{k} = \frac{(A_0^1 + A_1^1)}{\sqrt{2}}, \quad Z_{k} = \frac{(A_0^1 - A_1^1)}{\sqrt{2}}, \quad k \in \{1,4,7\}, \\
    &X_{k} = A_0^{k}, \quad Z_{k} = A_1^{k}, \quad k \in [9] / \{1,4,7\}.
\end{align*}

Plugging all of them inside $I_{9}' = 2\alpha_0 \tilde{P} + 2\sum_{i=1}^{9}\alpha_i \tilde{S}_i - \sum_{i=1}^{9}\alpha_i \tilde{S}_i^2 - \alpha_0 \tilde{P}^2$, we obtain $I_9'$ in terms of various monomials $A$'s.

\section{}\label{app:ap3}
Here we review the final ingredient proof of \cite{baccari2020device, baccari2020scalable}. After all the hard work, the final proof is somewhat tautological. We include it for completeness.

Given the above anticommutation relations for operators acting on all sites, we can now use Lemma \ref{lem:qubit} to introduce a set of local unitary operations $U_i: \mathcal{H}_{P_i} \to \mathcal{H}_{P_i}$ such that:
\begin{eqnarray}
U_iZ_iU_i^{\dagger} = \sz_i \otimes \si, \quad U_iX_iU_i^{\dagger} = \sx_i \otimes \si
\end{eqnarray}
for $i = 1, 2, \ldots, n.$

We define $\ket{\tilde{\psi}} = (\prod_{i=1}^{n} U_i \otimes \si)\ket{\psi}$. Now, using $\tilde{S}_i \ket{\psi} = \ket{\psi}$ implies:
\begin{equation}\label{eq:c3}
S_i \otimes \si \ket{\tilde{\psi}} = \ket{\tilde{\psi}}, \quad \forall i.
\end{equation}
Finally, we claim that the most general form of a state is:
\begin{equation} \label{eq}
\ket{\tilde{\psi}} = c\ket{\psi_1} \otimes \ket{\xi_1} + \sqrt{1-c^2} \ket{\psi_2} \otimes \ket{\xi_2},
\end{equation}
where $\ket{\psi_1}$ and $\ket{\psi_2}$ are two orthonormal states spanning the code space, and $c \in [0,1].$

The whole Hilbert space is $(\mathbb{C}^2)^{n} \otimes \mathcal{H}_{PE}$. We then write the Schmidt decomposition of $\ket{\tilde{\psi}}$ with respect to the tensor product structure:
\begin{eqnarray}
\ket{\tilde{\psi}} = \sum_{j} \lambda_j \ket{\eta_j} \ket{\phi_j},
\end{eqnarray}
where $\ket{\eta_j} \in (\mathbb{C}^2)^n$ and $\ket{\phi_j} \in \mathcal{H}_{PE}$ are orthogonal vectors in their respective Hilbert spaces. Now, eq \ref{eq:c3} forces $\ket{\eta_j}$ to become the codeword. In the case $\alpha_0 > 0$, it further reduces the state to the particular codewords. Now let us assume $\alpha_0 = 0$. Therefore, we can write $\ket{\eta_j} = c_j \ket{\psi_1} + d_j \ket{\psi_2}$ for some complex numbers $c_j, d_j$ satisfying $|c_j|^2 + |d_j|^2 = 1$. $\ket{\psi_1}$ and $\ket{\psi_2}$ are orthonormal states of the particular codespace. Hence, we can write the whole state as:
\begin{equation}
\ket{\tilde{\psi}} = c\ket{\psi_1} \otimes \ket{\xi_1} + \sqrt{1-c^2} \ket{\psi_2} \otimes \ket{\xi_2},
\end{equation}
which completes the proof of eq \ref{eq}.

\section{}\label{app:app4}

We consider the Hilbert space composed of $N$ parties (qubits) $(\mathbb{C}^2)^{\otimes N}$. Given a partition $Q \subset [N]$, $\bar{Q}=[N] \setminus Q$, we say that $\ket{\psi}$ is entangled across $Q|\bar{Q}$ if it cannot be written as $\ket{\psi}=\ket{\psi_Q} \otimes \ket{\psi_{\bar{Q}}}$ for some pure states $\ket{\psi_Q}$ and $\ket{\psi_{\bar{Q}}}$ supported on $Q$ and $\bar{Q}$, respectively. We then call $\ket{\psi}$ genuinely multipartite entangled (GME) if it is entangled across any nontrivial bipartition $Q|\bar{Q}$.

Consider a nontrivial subspace $V \subset (\mathbb{C}^d)^{\otimes N}$. We call $V$ GME \cite{demianowicz2018unextendible} (or, shortly, genuinely entangled) if all pure states belonging to it are GME. \cite{makuta2021self} gives a simple criterion-checking stabilizer subspaces to be genuinely entangled and upper bounds the genuinely entangled stabilizer subspaces of maximal dimension (Theorem 3). The bound is also shown to be tight by providing a particular example of a GME stabilizer subspace.

Additionally, \cite{makuta2021self} constructs such Bell inequalities for stabilizer subspaces that are genuinely entangled and of maximal dimension. For $N=7$, the maximal dimension is $8$. Therefore, this method cannot provide the code space for the Bell inequality of $[[7,1,3]]$.

\end{document}